\documentclass{IEEEtran4PSCC}
\ifCLASSINFOpdf
   \usepackage[pdftex]{graphicx}
\else
   \usepackage[dvips]{graphicx}
\fi
\usepackage[cmex10]{amsmath}
\usepackage{amssymb}
\usepackage{amsmath}

\usepackage{booktabs, siunitx}
\usepackage{multirow}

\usepackage{subfig}

\newtheorem{theorem}{Theorem}[section]

\newcommand{\qed}{\hfill$\blacksquare$}
\newenvironment{proof}{\noindent\textbf{Proof}\hspace*{1em}}{\qed\\}

\def\defeq{\triangleq} 

\usepackage[capitalize]{cleveref}
\crefname{assumption}{Assumption}{Assumptions}
\crefname{equation}{Eq.}{Eqs.}
\crefname{corollary}{Cor.}{Cors.}
\crefname{lemma}{Lem.}{Lems.}
\crefname{theorem}{Thm.}{Thms.}
\crefname{proposition}{Prop.}{Props.}
\crefname{assumption}{Assump.}{Assumps.}
\crefname{section}{Sec.}{Secs.}
\crefname{appendix}{App.}{Apps.}

\makeatletter
\let\old@ps@headings\ps@headings
\let\old@ps@IEEEtitlepagestyle\ps@IEEEtitlepagestyle
\def\psccfooter#1{%
    \def\ps@headings{%
        \old@ps@headings%
        \def\@oddfoot{\strut\hfill#1\hfill\strut}%
        \def\@evenfoot{\strut\hfill#1\hfill\strut}%
    }%
    \def\ps@IEEEtitlepagestyle{%
        \old@ps@IEEEtitlepagestyle%
        \def\@oddfoot{\strut\hfill#1\hfill\strut}%
        \def\@evenfoot{\strut\hfill#1\hfill\strut}%
    }%
    \ps@headings%
}
\makeatother

\begin{document}
%
\title{Estimating Technical Loss without Power Flows:\\ A Practical, Data-Driven Approach for Loss Estimation in Distribution Grids}

\author{
\IEEEauthorblockN{Mohini Bariya$^{\dagger}$, Genevieve Flaspohler$^{*\dagger}$}
\IEEEauthorblockA{$^*$\textit{n}Line, Inc., $^{\dagger}$Rhiza Research\\
Berkeley, USA\\
\{mohini, genevieve\}@rhizaresearch.org}
}


\maketitle

\begin{abstract}
Electric grids in low- and middle-income countries (LMICs) across the world face an acute challenge. To support global decarbonisation efforts and raise millions from energy poverty, these grids must shoulder substantial load growth while integrating distributed renewable generation. However, decades of rapid and poorly funded infrastructure expansions have led to national grids in many LMICs that are strained and weak, composed of aging, faulty, and undersized infrastructure. A cause and symptom of this weakness is excessive technical loss within ithe grid infrastructure during energy delivery, particularly at the distribution level; network losses are regularly estimated to be well over 20\%, compared to a baseline of 5\% in higher-income nations. Addressing technical loss through targeted interventions is essential for bolstering grids’ physical and economic strength. Unfortunately, current approaches for estimating and localizing technical loss require expensive, extensive power flow sensing, which is essentially absent in LMIC distribution systems. We present a novel approach to technical loss estimation without power flows, which leverages more readily available voltage magnitude measurements at sparse locations in the grid. This estimator puts loss estimation and localization within reach for LMIC grids globally, and provides a critical tool for the effective design, implementation, and evaluation of loss-reduction interventions.


\end{abstract}


\thanksto{\noindent This is a preprint of work supported by nLine, Inc.}

\section{Introduction}
Across low- and middle-income countries (LMICs), electric grids are central to planned economic growth and energy transition pathways. They must support the significantly larger electricity flows necessary to power economic activity and meet heightened customer demand, while handling the greater spatial and temporal variability introduced by renewable sources and end-use electrification. Recent analyses project 30-60\% electricity demand growth in LMICs \cite{iea2022world}. Yet, grids in LMICs are \textit{already} strained and weak, composed of aging, faulty and undersized equipment that cannot handle the current level and nature of demand and suffering from acute challenges in reliability and power quality that have significant human and economic consequences \cite{maruyama2019underutilized}.

A symptom and cause of this grid weakness is technical loss---energy lost within the electrical infrastructure itself due to Joule heating. In many LMICs, technical losses are estimated to be well over 20\% of total generation \cite{bhatti2015electric, jimenez2014sizing}\footnote{Without direct measurements of technical loss, disentangling total loss into technical and commercial losses is difficult. Still, staggering levels of total loss and clear evidence of infrastructure weakness attest to highly elevated technical loss in LMICs.}, compared to around 5\% in the United States and much of Western Europe \cite{eia, vasconcelos2008survey}. High technical losses have serious societal and environmental costs. Losses threaten the financial viability of electric utilities, as the cost of lost energy is too large to recoup through tariffs that many customers already struggle to afford. When generation is fossil fuel based, technical loss means heightened emissions per unit of electricity delivered. Weak grids, which lead to high loss, also manifest service quality issues such as frequent outages and poor power quality, which plague consumers across LMICs and severely curtail the economic and social benefits of electrification \cite{maruyama2019underutilized}. Finally, losses threaten successful climate mitigation and adaptation strategies that depend on strong grid infrastructure.

Reducing technical losses under the tight resource constraints faced by LMIC grids is challenging. The vast majority of technical loss occurs in the sprawling distribution network \cite{jimenez2014sizing}. It is impossible to upgrade this infrastructure wholesale, so targeted interventions are needed. The worst performing feeder in Accra, Ghana, for example, has around two orders of magnitude greater loss rates than the best performing feeder (as estimated by \textit{n}Line's large grid monitoring project in Ghana \cite{klugman2021watching}); this variation in performance within the distribution system underscores that targeting can maximize intervention efficiency. Effectively targeting loss-reduction activities demands knowledge of existing loss levels at high spatial and temporal resolution.

However, there is currently no practically feasible method for large-scale, high-resolution measurement or estimation of technical loss in LMIC distribution networks. Many existing methods demand comprehensive power flow sensing composed of conventional line sensorsi and smart metering to collectively capture all power flows into and out of the network \cite{meffe2009technical, anumaka2012analysis, hong2010calculating, neto2013probabilistic}. Other techniques have even higher sensing demands, requiring additional specialized monitoring to improve loss estimates, such as temperature sensors \cite{henriques2020monitoring}. Another class of methods entails system simulation, which requires detailed system models and, in lieu of direct power measurements, accurate load estimates \cite{abdulkareem2021investigating, queiroz2012energy, shenkman1990energy}. Neither extensive sensing nor simulation-ready models exist at scale for LMIC distribution networks, making both classes of approaches of limited applicability. The most practically useful prior work in LMIC contexts are engineering ``rules of thumb'', some of which have been enhanced to incorporate available measurements. For example, \cite{hong2010calculating} improves a well-established rule of thumb by using current measurements. However, rules of thumb were developed on data from distribution networks in the United States or Europe, often through regressing on the results of power flow simulations \cite{chen1994development}. Their generalizability, especially to LMIC contexts, is limited \cite{ibrahim2017system}.

In this work, we propose an alternative. We show mathematically that, under reasonable assumptions, voltage magnitude measurements alone can be used to estimate technical loss fractions across lines without power flow or load measurements. We present a mathematical correction step that mitigates deterioration of the loss estimate when sensor coverage is increasingly sparse. Finally, we discuss how leveraging sparse voltage measurements enables a practical sensing approach for technical loss, free from the comprehensive coverage requirements, privacy concerns \cite{mckenna2012smart}, sensor placement constraints, and installation challenges associated with power flow measurements. Our novel voltage-based loss estimator is termed \textit{\texttt{voss}}. We evaluate \texttt{voss} in simulation and describe how a low-cost, outlet-level voltage sensing approach---such as that deployed at scale by \textit{n}Line in Accra, Ghana---enables this method in practice.

The rest of this paper is organized as follows. \cref{sec:voss} derives and discusses the \texttt{voss} estimator. \cref{sec:sims} shares estimator results in simulation on IEEE test feeders. \cref{sec:real_world} discusses the practical considerations, strengths, and limitations of applying the estimator in the real world.

A note on notation: bold ($\mathbf{x}$) denotes complex (phasor) quantities while light ($x$) denotes scalars.

\section{\textbf{VOSS}: Voltage to Technical Loss}\label{sec:voss}
\begin{figure}[t]
\centering
\includegraphics[width=0.44\textwidth]{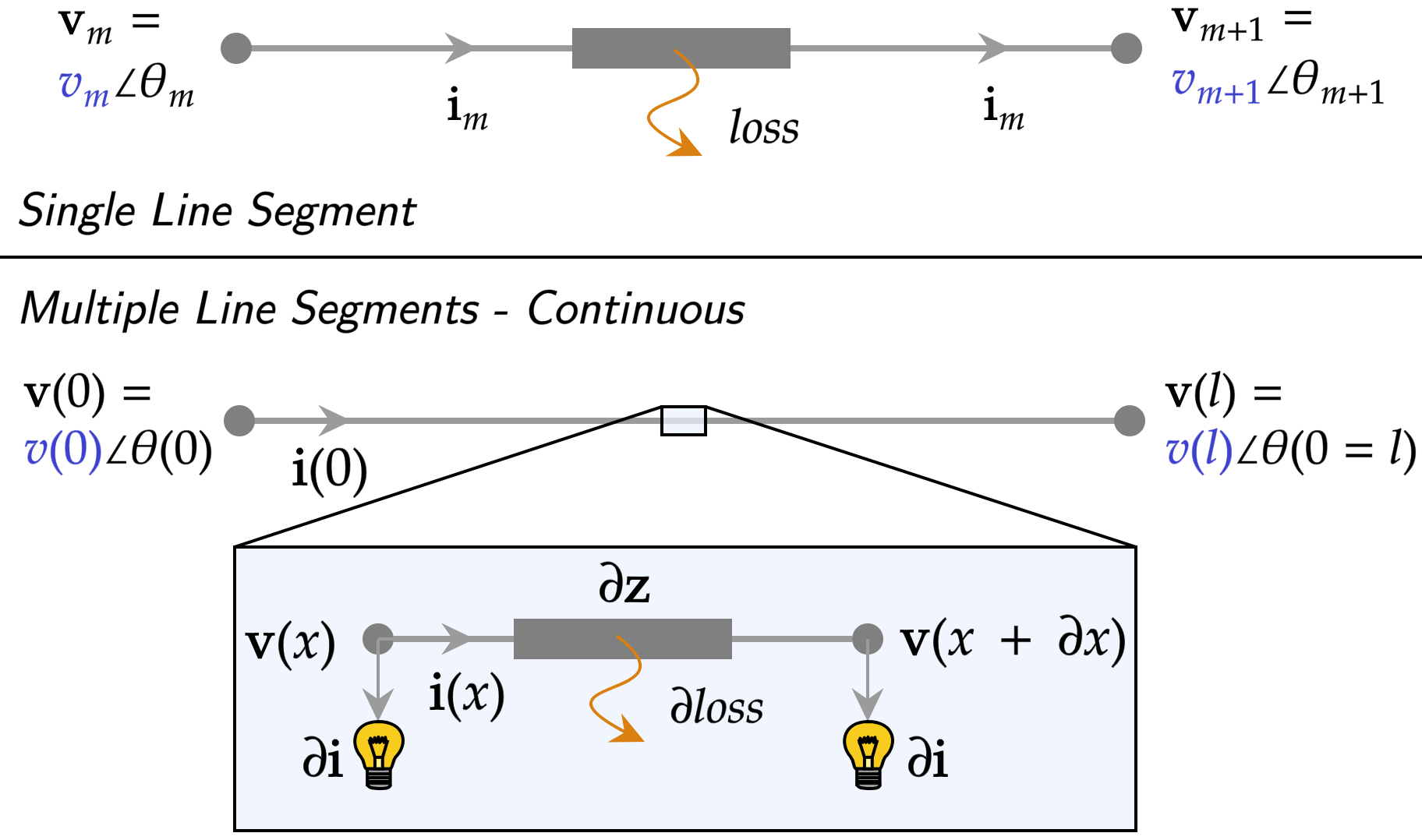}
\caption{The models for deriving the single segment (top) and mult-segment (bottom) \texttt{voss} estimator of \cref{sec:single_seg} and \cref{sec:multi_seg} respectively. Measured quantities are indicated in blue.}
\label{fig:model}
\end{figure}

The \texttt{voss} estimator estimates the fractional technical loss along a single network line using nodal voltage magnitude measurements at the two ends. The line can consist of a single segment---in which case no load couplings are present between a pair of nodal voltage measurements---or multiple segments, in which case load currents enter or exit the network between the pair of voltage measurements. We start by deriving the \texttt{voss} estimator for the simpler single segment scenario, and then proceed to the multiple segment case. Both scenarios, and corresponding notation, are visualized in \cref{fig:model}.

\subsection{Single Line Segment}\label{sec:single_seg}
We first derive an approximation of the single line segment loss, assuming no intervening load connections between the start and end point of the line. We seek an approximation that can be expressed in terms of voltage magnitudes at the start and end of the segment alone. 

\begin{theorem}
\label{thrm:single-segment}
Let the fraction of loss occurring on a single line segment $m$ with input and output power $\mathbf{s}_m$ and $\mathbf{s}_{m+1}$ respectively be denoted as: $loss_{frac}(m) \defeq \bigg|\frac{\textbf{s}_m-\textbf{s}_{m+1}}{\textbf{s}_m}\bigg|$. The line loss on segment $m$ can be approximated as:
\begin{equation}
 loss_{frac}(m) \approx \texttt{voss}(m) \defeq \frac{v_m^2-v_mv_{m+1}}{v_m^2},
 \label{eq:voss_ss}
\end{equation}
where $v_m$ and $v_{m+1}$ are the nodal voltage magnitudes at the start and end of the line respectively.
\end{theorem}

\begin{proof}
We wish to estimate the technical loss as a fraction of power input to line segment $m$. The phasor voltages at the two ends of the line are $\mathbf{v}_m = v_m\angle \theta_m$ and $\mathbf{v}_{m+1} = v_{m+1}\angle \theta_{m+1}$, and the current flowing in and out (as there are no intervening load connections) is $\mathbf{i}_m$. Then the input power is $\textbf{s}_m = \textbf{v}_m\textbf{i}_m^*$ and the output power is $\textbf{s}_{m+1} = \textbf{v}_{m+1}\textbf{i}_m^*$. This allows us to write the fraction of input power lost along the line as:
\begin{align*}
    loss_{frac}(m) = \bigg|\frac{\textbf{s}_m-\textbf{s}_{m+1}}{\textbf{s}_m}\bigg| = \bigg|\frac{\textbf{i}^*_m(\textbf{v}_m-\textbf{v}_{m+1})}{\textbf{v}_m\textbf{i}_m^*}\bigg|\\
    = \bigg|\frac{v_m^2 - \textbf{v}_m^*\textbf{v}_{m+1}}{v_m^2}\bigg| = \bigg|\frac{v_m^2-v_mv_{m+1}e^{j(\theta_m-\theta_{m+1})}}{v_m^2}\bigg|
\end{align*}
In many practical sensing contexts, only voltage \textit{magnitudes}---$v_m$ and $v_{m+1}$---are measured. We can apply the small angle approximation to the angle difference across the line \cite{bolognani2015existence}:
\begin{align}
    e^{j(\theta_m-\theta_{m+1})} \approx 1
\end{align}
This approximation allows the loss fraction to be expressed in terms of the available magnitude measurements. This is the \texttt{voss} estimate:
\begin{align}
    loss_{frac}(m) \approx \texttt{voss}(m) = \frac{v_m^2-v_mv_{m+1}}{v_m^2}
\end{align}
\end{proof}

The approximation of line loss in \cref{thrm:single-segment} is both intuitive---the amount of loss on a line is directly proportional to the drop in voltage magnitude that occurs along the line---and practically significant: the fraction of technical loss along a line can be estimated using easy-to-measure voltage magnitudes at either end of the line, and does not require a more challenging power flow measurement.

\subsection{Multiple Line Segments}\label{sec:multi_seg}
A more realistic scenario is one in which we have sparser sensing: rather than voltage measurements available between every pair of load coupling points, measurements are available at two ends of a multi-segment line with several intervening load connections. We are interested in identifying and compensating for the error in the \texttt{voss} estimator due to use of endpoint voltages with intermediate current (and power) extractions. 

\begin{theorem}
\label{thrm:multi-segment}
Let the fraction of loss occurring on a line segment of length $l$ be denoted as $loss_{frac}(0, l)$. The line loss along this segment can be approximated as:
\begin{equation}
 loss_{frac}(0, l) \approx \texttt{voss}(0, l) \defeq  \hat{c} \cdot \frac{v(0)^2-v(0)v(l)}{v(0)^2}
 \label{eq:voss_corrected}
\end{equation}
where $v(0)$ and $v(l)$ are the nodal voltage magnitudes at the start and end of the line respectively, and the correction factor $\hat{c}$ is defined as:
\begin{align}
    \hat{c} \triangleq 1-\bigg(\frac{\rho_v-\rho_s}{\rho_v+\rho_s}\bigg)\bigg(\frac{\rho_v+2\rho_s}{3\rho_v}\bigg),
    \label{eq:final_factor}
\end{align}
with parameters $\rho_s$, the ratio of the real power flowing into and out of the line, and $\rho_p$, the ratio of voltage magnitudes at the line ends.
\end{theorem}

\begin{proof} Both for simplicity and because we assume precise information on the spacing and size of loads is unavailable, we model this situation as highly symmetric. In the discrete case, this would mean identical loads separated by identical impedances, while in the continuous case we model a uniform impedance line with uniform current leakage along it's length. Bear in mind that our purpose is not to exactly model arbitrary network configurations, but to derive a workable estimator that demands minimal measurement coverage and network information. 

Our notation changes slightly for the continuous setup. The variable $x$ denotes distance along the line, which has a total length of $l$. We assume constant impedance and current leakage per unit of line length: 
\begin{align}
    \frac{\partial \mathbf{i}}{\partial x} = \iota,\ \frac{\partial \mathbf{z}}{\partial x} = \boldsymbol{\zeta}
\end{align}
The current flowing through the line as a function of distance along the line is:
\begin{align}
    \mathbf{i}(x) = i(0) - \int_0^x \iota \partial x = i(0) - \iota x
\end{align}
where we have chosen $i(0)$ to be real (such that it defines the reference angle), and simplified $\frac{\partial \mathbf{i}}{\partial x}$ to be real.
The voltage drop along the length of the line is:
\begin{align*}
    \mathbf{v}(0) - \mathbf{v}(l) = \int_0^l \boldsymbol{\zeta}\mathbf{i}(x) \partial x\\ = \int_0^l\boldsymbol{\zeta}(i(0)-\iota x)\partial x
    = \boldsymbol{\zeta}(i(0)l-\frac{\iota}{2}l^2)
\end{align*}
Solving for $\boldsymbol{\zeta}$ gives:
\begin{align*}
    \boldsymbol{\zeta} = \frac{\mathbf{v}(0)-\mathbf{v}(l)}{i(0)l-\frac{\iota}{2}l^2}
\end{align*}
The total power lost along the multi-segment line, denoted $\boldsymbol{\lambda}_{ms}$, is: 
\begin{align*}
    \boldsymbol{\lambda}_{ms} = \int_0^l \mathbf{i}^2(x) \boldsymbol{\zeta} \partial x = \boldsymbol{\zeta}\int_0^l(i(0)^2-2i(0)\iota x+\iota^2x^2)\partial x\\
    =\boldsymbol{\zeta}\bigg(i(0)^2l-i(0)\iota l^2+\frac{\iota^2}{3}l^3\bigg)\\
    = \boldsymbol{\zeta}\bigg(i(0)^2l - \frac{\iota l^2}{2}i(0)\bigg) + \boldsymbol{\zeta}\bigg(-\frac{\iota l^2}{2}i(0)+\frac{\iota^2l^3}{3}\bigg)\\
    = (\mathbf{v}(0)-\mathbf{v}(l))i^*(0)+(\mathbf{v}(0)-\mathbf{v}(l))i^*(0)\bigg(\frac{-\frac{\iota l}{2}+\frac{\iota^2l^2}{3i(0)}}{i(0)-\frac{\iota}{2}l}\bigg)
\end{align*}
$(\mathbf{v}(0)-\mathbf{v}(l))i^*(0)$ would be the loss across this line \textit{if} it consisted of a single segment. This is also the loss estimated by the single-segment \texttt{voss} estimator of \cref{eq:voss_ss}. We denote this quantity $\boldsymbol{\lambda}_{ss}$. Note that $\iota l$ is the total current that escapes along the line. We will express it as a fraction $\rho$ of the input current: $\iota l = \rho i(0)$. Plugging in these new definitions, we obtain: 
\begin{align*}
    \boldsymbol{\lambda}_{ms} = \boldsymbol{\lambda}_{ss}\bigg(1 + \frac{-\frac{\rho i(0)}{2}+\frac{\rho^2i(0)}{3}}{i(0)-\frac{\rho i(0)}{2}}\bigg) \\ = \boldsymbol{\lambda}_{ss}\bigg(1 - \rho\frac{3-2\rho}{6-3\rho}\bigg)
\end{align*}
We've derived a correction factor in terms of the fraction of current lost on the line, to obtain the true multi-segment loss $\boldsymbol{\lambda}_{ms}$ from the single segment equation, assuming the symmetric line model.
\begin{align}
    c \triangleq 1 - \rho\frac{3-2\rho}{6-3\rho} \label{eq:cfactor}
\end{align}
Note that the correction factor $c \leq 1$. We can apply this correction factor as is, using an engineering estimate of $\rho$ to update our single segment \texttt{voss} estimate for the multi-segment case. 

However, current loss is not a physically intuitive quantity and is difficult to estimate. We can re-express $c$ in more practical terms based on the measured voltage ratios, and an estimate of the \textit{power loss} fraction along the line. First note that we can rewrite $\rho$ as follows:
\begin{align*}
    \rho = \frac{i(0)-i(l)}{i(0)} = 1 - \bigg(\frac{\mathbf{s}(l)}{\mathbf{s}(0)}\bigg)^*\bigg(\frac{\mathbf{v}(0)}{\mathbf{v}(l)}\bigg)^*
\end{align*}
where $\mathbf{s}(l)$ and $\mathbf{s}(0)$ are the current flows into and out of the line respectively. Then, $\rho$ can be approximated as:
\begin{align*}
    \rho \approx 1 - \frac{p(l)}{p(0)}\frac{v(0)}{v(l)} \triangleq 1 - \frac{\rho_s}{\rho_v}
\end{align*}
where $\rho_s$ is ratio of between real power flowing in and out of the line, while $\rho_v$ is the ratio of voltage magnitudes at the line ends. This definition of $\rho$ is practically usable: $\rho_v$ can be computed from the available voltage magnitude measurements, while $\rho_s$ can be reasonably chosen through an engineering estimate. For example, if the line of interest starts at the feeder head and ends $\frac{1}{3}$ of the way down, we can reasonably estimate $\rho_s = \frac{2}{3}$, assuming load is uniformly distributed along the feeder length. Then, the refined \texttt{voss} estimator for the multi-segment case is: 
\begin{align}
    \hat{c} \triangleq 1-\bigg(\frac{\rho_v-\rho_s}{\rho_v+\rho_s}\bigg)\bigg(\frac{\rho_v+2\rho_s}{3\rho_v}\bigg)\\
    \texttt{voss}(0, l) = \hat{c} \cdot \frac{v(0)^2-v(0)v(l)}{v(0)^2}
\end{align}
As we expect $0 \leq \rho_s < \rho_v \leq 1$, the correction factor in \cref{eq:final_factor} should be $\leq 1$; indeed, since the un-approximated correction $c$ in \cref{eq:cfactor} is always $\leq 1$, if $\hat{c}$ is not, it suggests our approximations are inappropriate for the given context.
\end{proof}

\section{Simulation Results}\label{sec:sims}
We apply the \texttt{voss} estimator to two radial, unbalanced three-phase distribution IEEE test networks: the 13 node and 34 node test feeders \cite{kersting1991radial}. Measurements of nodal voltages and line flows are generated by simulating the networks in OpenDss \cite{montenegro2012real}.

\begin{table}[h!]
\centering
\begin{tabular}{llllll}
    \toprule
    Line & {} & Single segment & $\hat{c}$ & Multi-segment & True Loss\\
    \midrule
    \multirow{3}{*}{$800 - 814$} & A & 0.28 & 1.00 & 0.28 & 0.28\\
    {} & B & 0.20 & 0.96 & 0.19 & 0.19\\
    {} & C & 0.19 & 0.98 & 0.18 & 0.18\\
    \midrule
    $816 - 822$ & A & 0.010 & 0.75 & 0.073 & 0.078\\
    \midrule
    \multirow{3}{*}{$828 - 854$} & A & 0.059 & 0.99 & 0.058 & 0.058\\
    {} & B & 0.061 & 0.99 & 0.060 & 0.060\\
    {} & C & 0.056 & 0.98 & 0.055 & 0.055\\
    \bottomrule
\end{tabular}
\caption{Comparison of single and multi-segment voss estimate to true loss percentage on the IEEE 34 node test system.}\label{tab:correction}
\end{table}

\begin{figure*}[ht!]

\subfloat[Simulation results on IEEE 13 Node Test Feeder]{%
\includegraphics[width=0.95\textwidth]{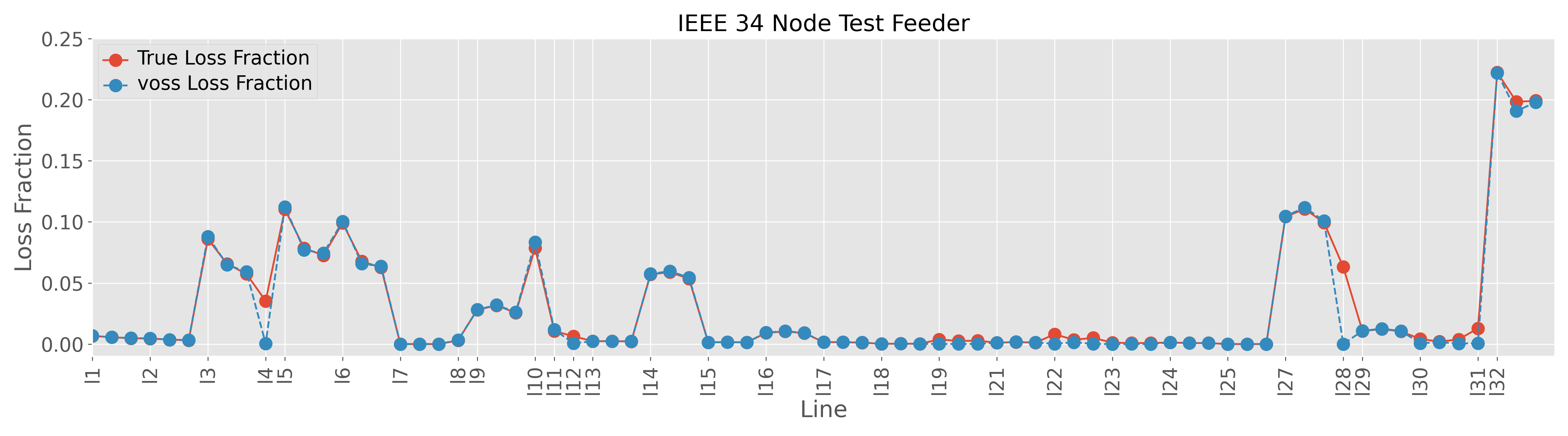}
}

\subfloat[Simulation results on IEEE 34 Node Test Feeder]{%
\includegraphics[width=0.95\textwidth]{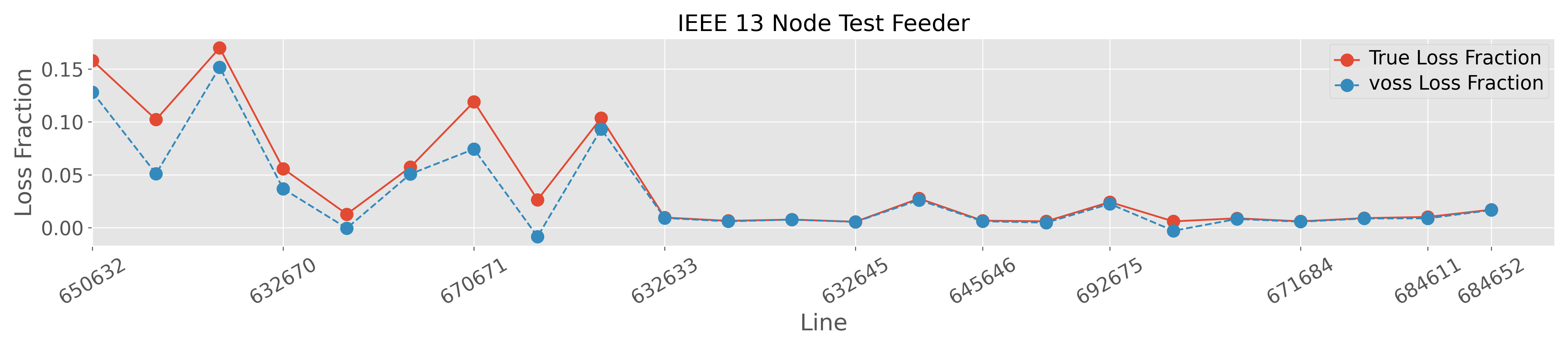}
}

\caption{Comparison of single-segment \texttt{voss} estimates in \cref{eq:voss_ss} to true loss fraction for individual line segments on two IEEE distribution test networks. Results are shown for each phase of each line. Lines \texttt{20} and \texttt{26} for the 34 node network and line \texttt{671680} for the 13 node network are excluded from the results as they have $\sim 0$ input power, rendering the loss fraction noise-sensitive and meaningless, especially for the purposes of estimator evaluation.}\label{fig:13and34}

\end{figure*}

On both feeders, we apply the single segment \texttt{voss} estimator from \cref{eq:voss_ss} to each line, comparing the result to the true line loss fraction. (Since these are three phase networks, each ``line'' consists of a line-phase pair, so we have a \texttt{voss} estimate for phase A of line 1, phase B of line 1, etc.) The results are plotted in \cref{fig:13and34} for the IEEE 34 node and 13 node test networks.

On the larger 34 node network, which has a few longer multi-segment lines, we also apply the multi-segment \texttt{voss} estimator of \cref{eq:voss_corrected}. Results are documented in \cref{tab:correction}, where we compare the uncorrected and corrected loss estimates (i.e. \cref{eq:voss_corrected} with and without $\hat{c}$).

These results compellingly demonstrate the efficacy of the \texttt{voss} estimator in simulation. In the next section, we discuss application of \texttt{voss} in operational grids.

\section{Real World Application}\label{sec:real_world}
The \texttt{voss} estimator has several practical advantages over existing technical loss estimation approaches. Many traditional approaches to evaluating loss in an electrical network are either purely simulation based or generic rules-of-thumb and thus estranged from any network measurements. Such approaches require detailed system information or are often too generalized to be accurate information on a given system context. At the other end of the algorithmic spectrum, we have methods that demand comprehensive power flow measurements throughout the system. Such measurements are expensive and logistically challenging to obtain. Meters installed on utility infrastructure, such as PMUs, are individually expensive and require specialized electricians as well as de-energization during installation \cite{overbye2010smart}. Smart meters, which measure endpoint consumption and could be used for loss estimation once coverage is comprehensive, have proven difficult to deploy at scale in LMICs \cite{hughes2020evaluation}. Even in industrialized nations with extensive smart metering deployments, using smart meters for grid insights has proved challenging due to privacy concerns and data access hurdles \cite{taft2012sensing}. Neither PMUs or smart meters will likely be practically viable as a sensing strategy to enable situational awareness (including loss estimation) in LMIC medium- and low-voltage networks.

\begin{figure*}[h!]
    \centering
    \includegraphics[width=\textwidth]{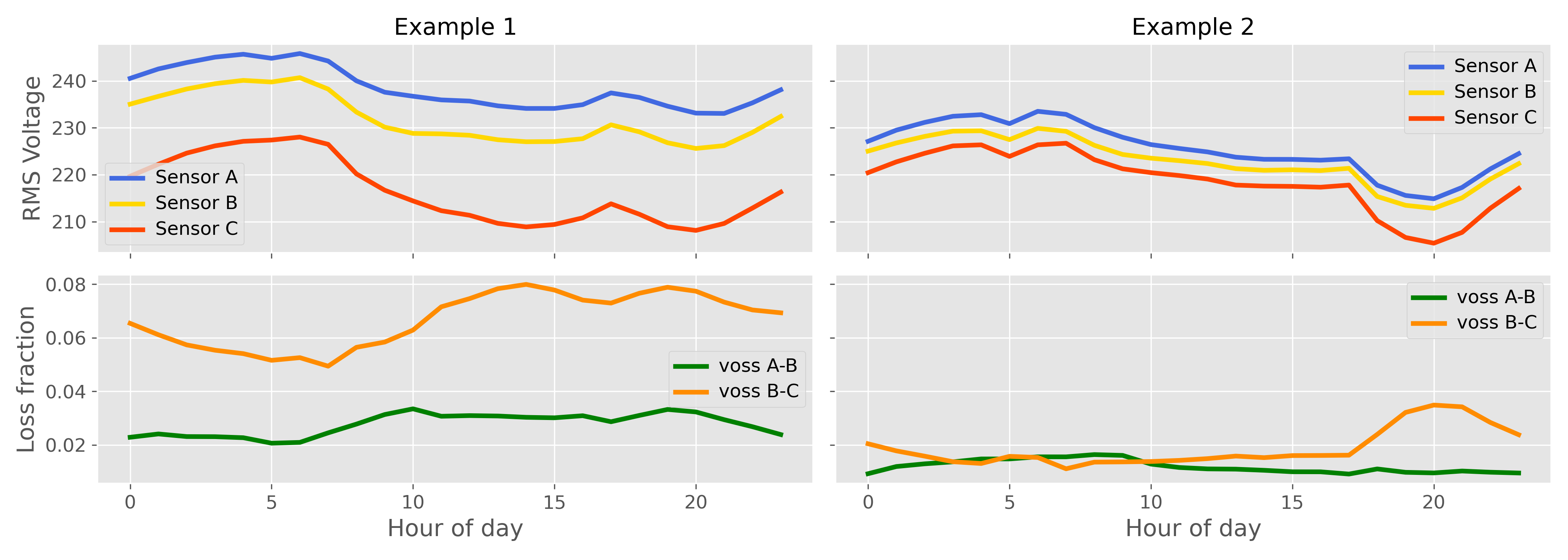}
    \caption{The \texttt{voss} estimator applied to real-world voltage magnitude measurements made \textit{n}Line's GridWatch sensors in Accra, Ghana. Each example is from a different part of Accra and uses measurements from three sensors situated linearly down a network, as demonstrated by the clear voltage drop between each sensor's voltage reading. For each example, we plot the daily voltage curve recorded by the sensor (top) and a corresponding fractional loss estimate along the line connecting sensors (bottom).}
    \label{fig:ghanaLoss}
\end{figure*}

Voltage measurements are considered less intrusive than consumption data, and therefore privacy constraints on their use for system monitoring are less onerous \cite{mckenna2012smart}. Voltage is a more global quantity than current or power. To obtain a flow measurement on a line, we must measure exactly the line of interest. To obtain a voltage measurement at a node, we can often get away with measuring voltage in the vicinity. This enables creative, affordable, and pragmatic sensing strategies for voltage-based loss estimation. One strategy developed by \textit{n}Line is the GridWatch system consisting of fleets of plug-point sensors that measure voltage at customer wall-outlets. GridWatch sensors are small, inexpensive, and extraordinarily simple to deploy: \textit{n}Line has deployed and maintained 1400+ such sensors across the city of Accra in Ghana over the past five years \cite{klugman2021watching}. Equipped with a universal SIM and battery backup, the sensors reliably report voltage magnitudes, frequency, and power state at two minute intervals.

While the voltage measurements reported by GridWatch sensors are not directly measured on the utility line, they can be used as an approximation of line voltage in order to compute a \texttt{voss} estimate. Given the significant voltage drops we have observed along network lines in LMIC contexts, this approximate but pragmatic approach to identifying the most lossy lines is likely adequate for a variety of targeting applications. The \texttt{voss} estimator can also be easily and intuitively refined when additional data is available. For example, when the input power to a feeder is known, the \texttt{voss} estimate of loss fraction can be converted into an estimate of power loss. 
In Fig. \ref{fig:ghanaLoss}, we provide an initial demonstration of how the \texttt{voss} estimator can be applied to real-world data. We use sample voltage magnitude measurements from our Ghana fleet. The selected sensors in Examples 1 and 2 are connected to low-voltage distribution networks that in turn are connected to medium-voltage lines in a linear pattern. Consequently, we can observe a voltage drop between sensors that likely corresponds to voltage drop along the medium voltage network. We can apply the \texttt{voss} estimator to these measurements to estimate loss fractions along the medium-voltage lines lying above the sensors; note that we can neglect the transformer turns ratio between the medium and low voltage networks by assuming it is constant across transformers and therefore is cancelled out in the estimator equation.

We do not have ground truth information about losses in Accra; however, the \texttt{voss} estimator produces reasonable numbers of medium-voltage line loss that are relatively high by the standards of industrialized countries, but match high estimates of total system loss in LMICs.

\section{Conclusion}
In this paper, we describe a novel approach to estimate and localize grid technical losses without power flow measurements, using only sparsely measured voltage magnitudes at the starts and ends of lines. That such a method can be rigorously derived is perhaps surprising and deeply empowering, enabling the development of practical, data-driven approaches to loss estimation with minimal, accessible sensing and modeling requirements. The novel \texttt{voss} estimator can provide essential insights on technical losses to electric grids in greatest need. These insights can inform effective, efficient system upgrades to prepare grids to support economic development and clean energy transitions.



%

\bibliographystyle{IEEEtran}
\bibliography{references}

\end{document}